\documentclass[12pt]{article}
\usepackage{amsmath}
\usepackage{amssymb,amsthm}
\usepackage[english]{babel}
\usepackage[textwidth=18.5cm,textheight=22cm]{geometry}
\newcommand{\C}{\mathbb{C}}

\newcommand{\N}{\mathbb{N}}

\renewcommand{\d}{\mathrm{d}}
\newcommand{\bt}{\boldsymbol{t}}

\newcommand{\bb}{\boldsymbol{\beta}}

\newcommand{\bu}{\boldsymbol{u}}

\newcommand{\bq}{\boldsymbol{q}}

\newtheorem{teh}{Theorem}

\begin{document}

\title{ Hodograph solutions of the dispersionless coupled KdV hierarchies, critical points and the
Euler-Poisson-Darboux equation }

\author{B. Konopelchenko $^{1}$, L. Mart\'{\i}nez Alonso$^{2}$ and E. Medina$^{3}$
\\
\emph{ $^1$ Dipartimento di Fisica, Universit\'a di Lecce and Sezione INFN}
\\ {\emph 73100 Lecce, Italy}\\
\emph{$^2$ Departamento de F\'{\i}sica Te\'orica II, Universidad
Complutense}\\
\emph{E28040 Madrid, Spain}\\
\emph{$^3$ Departamento de Matem\'aticas, Universidad de C\'adiz}\\
\emph{E11510 Puerto Real, C\'adiz, Spain}
}

\maketitle \abstract{It is shown that the hodograph solutions of the
dispersionless coupled KdV (dcKdV) hierarchies describe critical and
degenerate critical points of a scalar function which obeys the
Euler-Poisson-Darboux equation. Singular sectors of each dcKdV
hierarchy are found to be described by solutions of higher genus
dcKdV hierarchies. Concrete solutions exhibiting shock type
singularities are presented. }

\vspace*{.5cm}

\begin{center}\begin{minipage}{12cm}

\end{minipage}
\end{center}
\newpage

\section{Introduction}

In the present paper we study hierarchies of hydrodynamical systems
describing quasiclassical deformations of hyperelliptic curves
\cite{kod,bor}
\begin{equation}\label{cu}
p^2=u(\lambda),\quad u(\lambda):=\lambda^{m}-\sum_{i=0}^{m-1}\lambda^i\,u_i,\quad m\geq 1.
\end{equation}
These hierarchies are of interest for several reasons. First, there are hierarchies of important hydrodynamical type
 systems among them. For $m=1$ one has the Burgers-Hopf hierarchy \cite{whi,dubn}
 associated with the dispersionless  KdV equation $u_t=\dfrac{3}{2}\,u\,u_x$. For $m=2$ it is the hierarchy of higher equations
 for the 1-layer Benney
system (classsical long wave equation)
\begin{equation}\label{ben}
\begin{cases}
u_t+u\,u_x+v_x=0\\ \\
v_t+(u\,v)_x=0.
\end{cases}
\end{equation}
The system \eqref{ben} and the corresponding hierarchy  are
quasiclassical limits of the nonlinear Schr\"odinger (NLS) equation
and the NLS hierarchy \cite{zak}. For  $m\geq 3$ these hierarchies
turn to describe the singular sectors  of the above $m=1,2$
hierarchies \cite{kod}.

Second, all these hierarchies are the  dispersionless limits of integrable
coupled KdV (cKdV) hierarchies  \cite{lma}-\cite{for2} associated to
Schr\"odinger spectral problems
\begin{equation}\label{s1}
\partial_{xx}\,\psi=v(\lambda,x)\,\psi,\end{equation}
with potentials which are polynomials in the spectral parameter
$\lambda$
\[
\quad
v(\lambda,x):=\lambda^{m}-\sum_{i=0}^{m-1}\lambda^i\,v_i(x)\,\quad
m\geq 1,
\]
The cKdV hierarchies have been studied in \cite{lma}-\cite{for2}, they have bi-Hamiltonian structures and,
 as a consequence of this property, the dispersionless  expansions of their solutions  possess  interesting features
such as the quasi-triviality property \cite{dub1}-\cite{dub2}. Moreover, the cKdV hierarchies arise also in the study of the
singular sectors of the KdV and AKNS hierarchies \cite{man, borl}. Henceforth we  will refer to the  hierarchies of
 hydrodynamical systems  associated with the curves \eqref{cu} for a fixed $m$
 as the $m$-th dispersionless coupled KdV (dcKdV$_m$) hierarchies. The Hamiltonian structures of the dcKdV$_m$
 hierarchies have been studied in \cite{fer}. At last,  it should be noticed that the dcKdV$_m$ hierarchies
   are closely connected with the higher genus Whitham hierarchies introduced in \cite{kri}.

\vspace{0.3cm}

In our analysis of the hodograph equations for the dcKdV$_m$  hierarchies  we use Riemann invariants  $\beta_i$
 (roots  of the polynomial $u(\lambda)$ in \eqref{cu}) which provide a specially convenient system of coordinates.
We show that the dcKdV$_m$  hodograph equations have the form
\begin{equation}\label{crii}
\dfrac{\partial W_m(\bt,\bb)}{\partial\beta_i}=0,\quad i=1,\ldots,m,
\end{equation}
 where $\bt=(t_1,t_2,\ldots)$ are times of the hierarchy and
\begin{equation}\label{w}
W_m(\bt,\bb):=\oint_{\gamma}\dfrac{\d \lambda}{2\,i\,\pi}\,\dfrac{\sum_{n\geq 0}t_n\,\lambda^n}{\sqrt{\prod_{i=1}^{m}
(1-\beta_i/\lambda)}}.
\end{equation}
Here $\gamma$ denotes a large positively oriented circle  $|\lambda|=r$. Thus,
the hodograph solutions of the dcKdV$_m$  hierarchies describe critical points of the functions $W_m(\bt,\bb)$.
 These functions turn to be very special as they satisfy a well-known system of equations
  in differential geometry:  the Euler-Poisson-Darboux  (EPD) equations \cite{dar}
\begin{equation}\label{epd1i}
2\,(\beta_i-\beta_j)\, \dfrac{\partial^2
W_m}{\partial\beta_i\,\partial\beta_j}= \dfrac{\partial
W_m}{\partial\beta_i}-\dfrac{\partial W_m}{\partial\beta_j}.
\end{equation}
The system \eqref{epd1i}  has also appeared in the theory of the
 Whitham equations arising in the small dispersion limit  of the KdV equations  \cite{tia1}-\cite{gra},
  and in the theory of hydrodynamic chains \cite{pav}.

  We also  study  the singular sectors $\mathcal{M}_{m}^{\mbox{sing}}$  of the spaces of hodograph solutions for the dcKdV$_m$  hierarchies. They are  given by the points $(\bt,\bb)$ such that
 \begin{equation}\label{siei}
 \mbox{rank$\Big ( \dfrac{\partial^2 W_m(\bt,\bb)}{\partial \beta_i\, \partial \beta_j}\Big)<m$}.
 \end{equation}
The  varieties $\mathcal{M}_{m}^{\mbox{sing}}$ provide us with special classes of degenerate critical points of
  the function $W_m$ within the general theory of critical points developed by V. I. Arnold and others
   about fourty years ago
 \cite{arn1,arn2}.  The use of equations \eqref{crii}-\eqref{epd1i}  simplify
drastically the analysis of the structure of these
 singular sectors. In particular,  we prove that there is a nested
sequence of  subvarieties
\begin{equation}\label{fil0}
\mathcal{M}_{m}^{\mbox{sing}}\supset  \mathcal{M}_{m,1}^{\mbox{sing}} \supset \mathcal{M}_{m,2}^{\mbox{sing}}\supset \cdots\mathcal{M}_{m,q}^{\mbox{sing}}\supset \cdots,
\end{equation}
which  represents  subsets of the singular sector  $\mathcal{M}_{m}^{\mbox{sing}}$ of the dcKdV$_{m}$ hierarchy with increasing singular degree $q$, such that each
 $\mathcal{M}_{m,q}^{\mbox{sing}}$ is determined by a class of hodograph solutions of the dcKdV$_{m+2\,q}$ hierarchy.

 \vspace{0.3cm}

The paper is organized as follows. The dcKdV$_m$  hierarchies are described in Section 2. Equations (4)-(6) are derived
in Section 3. Section 4 deals with the analysis of the singular sectors of the dcKdV$_m$  hierarchies
in terms of their associated hodograph equations. The relation between singular points of the dcKdV$_m$ hodograph
 equations and solutions of higher dcKdV$_{m+2\,q}$ hodograph equations is stated in Section 4.
 Some concrete examples involving shock singularities of the Burgers-Hopf equation and the 1-layer Benney system are
 presented in Section 5.

\section{The dcKdV$_m$  hierarchies}

Given a positive integer $m\geq 1$ we consider the set $M_m$ of algebraic curves \eqref{cu}.
For  $m=2\,g+1$ (odd case) and $m=2\,g+2$ (even case) these curves are, generically, hyperelliptic  Riemann surfaces of genus $g$. We will denote  by $\bq=(q_1,\ldots, q_{m})$ any of the two sets of parameters $\bu:=(u_0,\ldots,u_{m-1})$ or $\bb:=(\beta_1,\ldots,\beta_{m})$ which determine  the curves \eqref{cu}
\begin{equation}\label{r11}
u(\lambda)=\lambda^{m}-\sum_{i=0}^{m-1}\lambda^i\,u_i=\prod_{i=1}^{m}(\lambda-\beta_i).
\end{equation}
 Obviously, for any fixed $\bb$ all the permutations $\sigma(\bb):=(\beta_{\sigma(1)},\ldots,\beta_{\sigma(m)})$ represent the same element of $M_m$. Note also that \begin{equation}\label{vi}
u_i=(-1)^{m-i-1}\mathrm{s}_{m-i}(\bb),
\end{equation}
where $\mathrm{s}_k$ are the elementary symmetric
polynomials
\[
\mathrm{s}_k=\sum_{1\leq i_1<\ldots<i_k\leq m}\beta_{i_1}\cdots \beta_{i_k}.
\]

We  next introduce the dcKdV$_m$ hierarchy  as a particular  systems of commuting flows
\[
\bq(\bt),\quad \bt:=(x:=t_0,t_1,t_2,\ldots),
\]
on $M_m$.
In order to define
these flows we use the set $\mathcal{L}$ of formal power series
\[
f(z)=\sum_{n=-\infty}^{+\infty}\,c_n\,z^n,
\]
where
\[
z:=\lambda^{1/2}\,\, \mbox{for $m=2\,g+1$};\quad
z:=\lambda\,\, \mbox{for $m=2\,g+2$}.
\]
For any given $m\geq 1$ a distinguished element of $\mathcal{L}$ is provided by the branch of $p=\sqrt{u(\lambda)}$ such that as    $z\rightarrow\infty$ has an expansion of the form
\begin{equation}\label{pes}
\begin{cases}
p(z,\bq)=z^{2\,g+1}\,\Big(1+\sum_{n\geq 1}\dfrac{b_n(\bq)}{z^{2\,n}}\Big)
,\quad m=2\,g+1,\\\\
p(z,\bq)=z^{g+1}\,\Big(1+\sum_{n\geq 1}\dfrac{b_n(\bq)}{z^{n}}\Big),
\quad m=2\,g+2.
\end{cases}
\end{equation}
We  define the following splittings $\mathcal{L}=\mathcal{L}_{(+,\,\bq)}\bigoplus\mathcal{L}_{(-,\,\bq)}$
\begin{equation}\label{pol}
 f_{(+,\,\bq)}(z):=\Big(\dfrac{f(z)}{p(z,\bq)}\Big)_{\oplus}\, p(z,\bq),
\quad f_{(-,\,\bq)}(z):=\Big(\dfrac{f(z)}{p(z,\bq)}\Big)_{\ominus}\, p(z,\bq),
\end{equation}
where $f_{\oplus}$ and $f_{\ominus}$ stand for the standard projections on positive and strictly negative powers of $z$, respectively
\[
f_{\oplus}(z):=\sum_{n=0}^N c_n\,z^n,\quad f_{\ominus}(z):=\sum_{n=-\infty}^{-1} c_n\,z^n.
\]

\vspace{0.3cm}

The  dcKdV$_m$  flows $\bq(\bt)$ are characterized by the following condition: There exists a family of functions
$S(z,\bt,\bq(\bt))$ in $\mathcal{L}$ satisfying
\begin{equation}\label{kdV}
\partial_{t_n}\, S(z,\bt,\bq(\bt))=\Omega_n(z,\bq(\bt)),
\quad n\geq 0.
\end{equation}
where
\begin{equation}\label{omegas}
\Omega_n(z,\bq):=(\lambda(z)^{n+m/2})_{(+,\,\bq)}=
\begin{cases}
(z^{2\,n+2\,g+1})_{(+,\,\bq)},\quad m=2\,g+1\\\\
(z^{n+g+1})_{(+,\,\bq)},\quad m=2\,g+2,
\end{cases}
\quad n\geq 0.
\end{equation}
We notice that
\begin{equation}\label{ere}
\Omega_n(z,\bq)=
\Big(\lambda^n\,R(\lambda(z),\bq)\Big)_{\oplus}\,p.
\end{equation}
where $R$ is the generating function
\begin{equation}\label{res}
R(\lambda,\bq):=\sqrt{\dfrac{\lambda^{m}}{u(\lambda)}}=\sum_{n\geq 0}
\dfrac{R_n(\bq)}{\lambda^n},\quad \lambda\rightarrow\infty.
\end{equation}
The coefficients $R_n(\bq)$ are polynomials in the coordinates $\bq$, for example
\[
R_0=1,\quad R_1=\dfrac{1}{2}\,u_{m-1},\quad R_2=\dfrac{1}{2}\,u_{m-2}
+\dfrac{3}{8}\,u_{m-1}^2,\quad \ldots \]
Functions $S$ which satisfy \eqref{kdV} will be referred to as \emph{action functions} of the dcKdV$_m$  hierarchy. This kind of
generating functions $S$ has been already used in the theory of dispersionless integrable systems (see e.g. \cite{kri}).
It can be proved \cite{kod} that \eqref{kdV} is a  compatible system of equations for $S$. In fact its general solution will
be determined in the next section.
We notice that for $n=0$ the equation \eqref{kdV} reads
\begin{equation}\label{pe}
\partial_x\, S(z,\bt,\bq(\bt))=p(z,\bq(\bt)),
\end{equation}
so that \eqref{kdV} is equivalent to the system
\begin{equation}\label{kdvp}
\partial_{t_n} p(z,\bq(t))=\partial_x\,\Omega_n(z,\bq(\bt)),\quad n\geq 0.
\end{equation}

We will henceforth refer to the  dcKdV$_m$ hierarchy for $m= 2\,g+1$ and $m=2\,g+2$   as the Burgers-Hopf  (BH$_g$) and the dispersionless Jaulent-Miodek (dJM$_g$) hierarchies, respectively. Observe that both hierarchies, BH$_g$ and dJM$_g$ determine deformations of hyperelliptic Riemann surfaces of genus $g$.  In our work we will always consider an arbitrary but finite number of these flows.

\vspace{0.3cm}

Since $u=u(\lambda(z),\bq)=p(z,\bq)^2$,  the operator $J=J(\lambda,u)$ defined by
\begin{align*}
J:&=2\,p\cdot\partial_x\cdot p=2\,u\,\partial_x+u_x,\\\\
J=&\sum_{i=0}^m\,\lambda^i\,J_i, \quad J_m=2\,\partial_x,\quad J_i=-(2\,u_i\,\partial_x+u_{i,x}),\quad u_m:=-1,
\end{align*}
satisfies $J\,R=0$. Then from \eqref{kdvp} it follows that
\begin{equation}\label{kdVu}
\partial_n\, u=
J\,\Big(\lambda^n\,R(\lambda,\bu)\Big)_{\oplus}=-J\,\Big(\lambda^n\,R(\lambda,\bu)\Big)_{\ominus},
\end{equation}
which constitutes the dcKdV$_m$  hierarchy in terms of  the coordinates $u_i$
\begin{equation}\label{kdVui}
\partial_n\, u_i=\sum_{l-k=i,\,k\geq 1}J_l\,R_{n+k}(\bu),\quad i=0,\ldots,m-1.
\end{equation}

From \eqref{kdvp} it also follows that
\[
\partial_{t_n}\,\log p(z,\bq)=
\dfrac{\partial_x\, \Big[\Big(\lambda(z)^n\,R(\lambda(z),\bq)\Big)_{\oplus}\,p\Big]}{p(z,\bq)},
\]
and then, identifying the residues of both sides at $\lambda=\beta_i$,
we get
\begin{equation}\label{bhb}
\partial_n\,\beta_i=\omega_{n,i}(\bb)\,\partial_x\,\beta_i,\quad i=1,\ldots,m,
\end{equation}
where
\begin{equation}\label{bhba}
\omega_{n,i}(\bb):=(\lambda^n\,
R(\lambda,\bb))_{\oplus}|_{\lambda=\beta_i}.
\end{equation}
The systems \eqref{bhb} are the equations of the dcKdV$_m$  hierarchy in terms of the coordinates $\beta_i$. Observe that we have two dcKdV$_m$  hierarchies, BH$_g$ and dJM$_g$, which determine deformations of hyperelliptic Riemann surfaces of genus $g$. It can be shown \cite{bor,fer}
that the dcKdV$_m$  flows  are bi-Hamiltonian systems.

\vspace{0.3 cm}

We next present some examples of interesting flows in the dcKdV$_m$ hierarchies.
The  dcKdV$_1$  hierarchy is
associated to the curve
\[
p^2-u(\lambda)=0,\quad u(\lambda)=\lambda-v,\quad v:=u_0=\beta_1.
\]
The corresponding  flows are given by
\[
\partial_{t_n}\,v=c_n\,v^n\, v_x,\quad c_n:=\dfrac{(2\,n+1)!!}{2^n\,n!},\quad n\geq 1,
\]
and constitute the Burgers-Hopf hierarchy BH$_0$.
In particular the $t_1$-flow is the Burgers-Hopf equation
\[
\partial_t\,v=\dfrac{3}{2}\,v\,v_x,
\]
which is in turn the dispersionless limit of the KdV equation.

\vspace{0.3cm}

The dcKdV$_2$ (dJM$_0$)  hierarchy is associated to the curve
\[
p^2-u(\lambda)=0,\quad u(\lambda)=\lambda^2-\lambda\, u_1-u_0=(\lambda-\beta_1)\,(\lambda-\beta_2),
\]
\[
u_1=\beta_1+\beta_2,\quad u_0=-\beta_1\,\beta_2.
\]
The $t_1$-flow of this hierarchy is given by the  disperssionless Jaulent-Miodek system
\begin{equation}\label{jm}
\begin{cases}
\partial_{t_1}\,u_0=u_0\,u_{1\,x}+\dfrac{1}{2}\,u_1\,u_{0\,x},\\ \\
\partial_{t_1}\,u_1=u_{0\,x}+\dfrac{3}{2}\,u_1\,u_{1\,x},
\end{cases}
\end{equation}
which under the changes of dependent variables
\[
u=-u_1,\quad
v=u_0+\dfrac{u_1^2}{4},
\]
becomes the 1-layer Benney system \eqref{ben}. In terms of the Riemann invariants $\beta_1$ and $\beta_2$
\[
u=-(\beta_1+\beta_2),\quad v=(\beta_1-\beta_2)^2/4,
\]
 the system \eqref{ben} takes the well-known form
\begin{equation}\label{benr}
\begin{cases}
\partial_{t_1}\,\beta_1=\dfrac{1}{2}\,(3\,\beta_1+\beta_2)\,\beta_{1\,x},\\ \\
\partial_{t_1}\,\beta_2=\dfrac{1}{2}\,(3\,\beta_2+\beta_1)\,\beta_{2\,x}.
\end{cases}
\end{equation}
For $v>0$ the 1-layer Benney system is hyperbolic while for $v<0$ it is elliptic.

\vspace{0.3cm}

Finally, we consider  the BH$_1$ hierarchy. Its associated curve  is given by
$$\everymath{\displaystyle}\begin{array}{l}
p^2\,-\,u(\lambda)\,=\,0,\quad
u(\lambda)\,=\,\lambda^3\,-\,\lambda^2\,u_2\,-\,\lambda\,u_1\,-\,u_0\,=\,(\lambda\,-\,\beta_1)
\,(\lambda\,-\,\beta_2)\,(\lambda\,-\,\beta_3),\\  \\
u_1\,=\,\beta_1\,+\,\beta_2\,+\,\beta_3,\quad
u_2\,=\,-\,(\beta_1\,\beta_2\,+\,\beta_1\,\beta_3\,+\,\beta_2\,\beta_3),\quad
u_3\,=\,\beta_1\,\beta_2\,\beta_3.
\end{array}$$
The first flow takes the forms
\begin{equation}\label{f3}
\everymath{\displaystyle}\left\{\begin{array}{l}
\partial_{t_1}u_0\,=\,\frac{1}{2}\,u_2\,u_{0\,x}\,+\,u_0\,u_{2\,x},\\  \\
\partial_{t_1}u_1\,=\,u_{0\,x}\,+\,\frac{1}{2}\,u_2\,u_{1\,x}\,+\,u_1\,u_{2\,x},\\  \\
\partial_{t_1}u_2\,=\,u_{1\,x}\,+\,\frac{3}{2}\,u_2\,u_{2\,x}.
\end{array}\right.
\Longleftrightarrow \quad
\left\{\begin{array}{l}
\partial_{t_1}\beta_1\,=\,\frac{1}{2}\,(3\,\beta_1\,+\,\beta_2\,+\,\beta_3)\,\beta_{1\,x},\\  \\
\partial_{t_1}\beta_2\,=\,\frac{1}{2}\,(\beta_1\,+\,3\,\beta_2\,+\,\beta_3)\,\beta_{2\,x},\\  \\
\partial_{t_1}\beta_3\,=\,\frac{1}{2}\,(\beta_1\,+\,\beta_2\,+\,3\,\beta_3)\,\beta_{3\,x}.
\end{array}\right.
\end{equation}

\section{Hodograph equations for  dcKdV$_m$  hierarchies and the Euler-Poisson-Darboux equation}

 Let us introduce the function
 \begin{equation}\label{efe}
W_m(\bt,\bq):=\oint_{\gamma}\dfrac{\d \lambda}{2\,i\,\pi}\,U(\lambda,\bt)\,
R(\lambda,\bq)=\sum_{n\geq0}\,t_n\,R_{n+1}(\bq),
\end{equation}
where $\gamma$ denotes a large positively oriented circle  $|\lambda|=r$ ,
$U(\lambda,\bt):=\sum_{n\geq 0}t_n\,\lambda^n$ and $R(\lambda,\bq)$ is the function defined in \eqref{res}.
\begin{teh}
If the functions $\bq(\bt)=(q_1(\bt,\ldots,q_m(\bt))$  satisfy the system of hodograph equations
\begin{equation}\label{hod}
\dfrac {\partial W_m(\bt,\bq)}{\partial q_i}=0,\quad i=1,\ldots,m,
\end{equation}
then $\bq(\bt)$ is a solution of the dcKdV$_m$  hierarchy.\end{teh}
\begin{proof}
We are going to prove that the function
\begin{equation}\label{sol}
S(z,\bt,\bq(\bt))=\sum_{n\geq 0}\, t_n\,\Omega_n(z,\bq(\bt))=\Big(U(\lambda(z),\bt)\,
R(\lambda(z),\bq(\bt))\Big)_{\oplus}\,p(z,\bq(\bt)),
\end{equation}
is an action function for  the dcKdV$_m$  hierarchy.
 By differentiating \eqref{sol} with respect to $t_n$ we have that
\begin{equation}\label{part}
\partial_n\,S=\Omega_n+(U\,\partial_n\,R)_{\oplus}\,p+
(U\,R)_{\oplus}\,\partial_n\,p,
\end{equation}
We now use the coordinates $\bb=(\beta_1,\ldots,\beta_m)$ so that we may take advantage of the identities
\begin{equation}\label{idd}
\partial_{\beta_i}\,p=-\dfrac{1}{2}\,\dfrac{p}{\lambda-\beta_i},
\quad
\partial_{\beta_i}\,R=\dfrac{1}{2}\,\dfrac{R}{\lambda-\beta_i}.
\end{equation}
Thus we deduce that
\begin{equation}\label{t1}
(U\,\partial_n\,R)_{\oplus}\,p+
(U\,R)_{\oplus}\,\partial_n\,p=\dfrac{1}{2}\,\sum_{i=1}^m\,\Big[\Big(\dfrac{U\,R}{\lambda-\beta_i}\Big)_{\oplus}
-\dfrac{(U\,R)_{\oplus}}{\lambda-\beta_i}\Big]\,p\,\partial_n\,\beta_i .
\end{equation}
On the other hand
\begin{equation}\label{simm}
\dfrac {\partial W_m(\bt,\bb)}{\partial \beta_i}=\dfrac{1}{2}\,\oint_{\gamma}\dfrac{\d \lambda}{2\,i\,\pi}\,\dfrac{U(\lambda,\bt)\,
R(\lambda,\bb)}{\lambda-\beta_i}=\dfrac{1}{2}\,\oint_{\gamma}\dfrac{\d \lambda}{2\,i\,\pi}\,\dfrac{(U(\lambda,\bt)\,
R(\lambda,\bb))_{\oplus}}{\lambda-\beta_i}.
\end{equation}
Hence the hodograph equations \eqref{hod} can be written as
\begin{equation}\label{hod1}
(U(\lambda,\bt)\,
R(\lambda,\bb(\bt)))_{\oplus}|_{\lambda=\beta_i}=0,\quad i=1,\ldots,m.
\end{equation}
Thus we have that $(U(\lambda,\bt)\,
R(\lambda,\bb(\bt))_{\oplus}$ is a polynomial in $\lambda$ which vanish at $\lambda=\beta_i(\bt)$ for all $i$. As a consequence
\[
\dfrac{(U\,R)_{\oplus}}{\lambda-\beta_i}=\Big(\dfrac{(U\,R)_{\oplus}}{\lambda-\beta_i}\Big)_{\oplus}=\Big(\dfrac{U\,R}{\lambda-\beta_i}\Big)_{\oplus}.
\]
Then from \eqref{part} and \eqref{t1} we deduce that $\partial_n\,S=\Omega_n$ and therefore the statement follows.

\end{proof}

Using \eqref{efe} we obtain that
 the hodograph equations \eqref{hod} can be expressed as
\begin{equation}\label{hodues}
\sum_{n\geq0}\,t_n\,\dfrac{\partial R_{n+1}(\bq)}{\partial q_i}=0,\quad i=1,\ldots,m.
\end{equation}
Furthermore, from \eqref{bhb}, \eqref{bhba} and \eqref{hod1} the
hodograph equations \eqref{hod} can be also written as \cite{kod}
\begin{equation}\label{hoduesss}
\sum_{n\geq0}\,t_n\,\omega_{n,i}(\bb)=0,\quad i=1,\ldots,m,
\end{equation}
which represent the hodograph transform for the dcKdV$_m$ hierarchy
of flows in hydrodynamic form.

  Notice also that we may shift the time parameters
$t_n\rightarrow t_n-c_n$ in \eqref{hodues} to get solutions
depending on an arbitrary number of constants.

\vspace{0.3cm}

It is easy to see that the generating function
\[
R(\lambda,\bb):=\sqrt{\dfrac{\lambda^{m}}{u(\lambda)}}=\sqrt{\dfrac{\lambda^{m}}{\prod_{i=1}^{m}(\lambda-\beta_i)}},
\]
is a symmetric solution of the EPD equation
\begin{equation}\label{epd}
2\,(\beta_i-\beta_j)\, \dfrac{\partial^2 R}{\partial\beta_i\,\partial\beta_j}=
\dfrac{\partial R}{\partial\beta_i}-\dfrac{\partial R}{\partial\beta_j}.
\end{equation}
Consequently, the same property is satisfied by $W(\bt,\bb)$ for all $\bt$. Thus, we have proved

\begin{teh}
The solutions $(\bt,\bb)$ of the hodograph equations
\begin{equation}\label{s0}
\dfrac{\partial W_m(\bt,\bb)}{\partial\beta_i}=0,\quad i=1,\ldots,m,
\end{equation}
are the critical points
 of the solution
\[
W_m(\bt,\bb):=\oint_{\gamma}\dfrac{\d \lambda}{2\,i\,\pi}\,\dfrac{U(\lambda,\bt)}{\sqrt{\prod_{i=1}^{m}
(1-\beta_i/\lambda)}}
\]
 of the  EPD equation
\begin{equation}\label{epd}
2\,(\beta_i-\beta_j)\, \dfrac{\partial^2
W_m}{\partial\beta_i\,\partial\beta_j}= \dfrac{\partial
W_m}{\partial\beta_i}-\dfrac{\partial W_m}{\partial\beta_j}.
\end{equation}
\end{teh}

\vspace{0.3cm}
Let us denote by
$\mathcal{M}_m$  the \emph{variety}  of points $(\bt,\bb)\in \C^{\infty}\times \C^m$  which satisfy the hodograph equations \eqref{s0}.
From \eqref{simm} it is clear
that for any permutation $\sigma$ of $\{1,\ldots,m\}$ the functions
\begin{equation}\label{efee}
 F_i(\bt,\bb):=\dfrac {\partial W_m(\bt,\bb)}{\partial \beta_i},
\end{equation}
satisfy
\begin{equation}\label{h1aa}
F_i(\bt,\sigma(\bb))=F_{\sigma(i)}(\bt,\bb).
\end{equation}
Then, it is clear
 that $\mathcal{M}_m$ is invariant under the action of the group of permutations
\[
(\bt,\bb)\in \mathcal{M}_m\Longrightarrow (\bt,\sigma(\bb))\in \mathcal{M}_m .
\]

\vspace{0.3cm}

 If $(\bt,\bb)$ is a solution of  \eqref{s0} such that $\beta_i\neq \beta_j$ for all $i\neq j$ then it will be called an \emph{unreduced} solution of  \eqref{s0}. In this case  the EPD equation \eqref{epd} implies that
 \begin{equation}\label{cros}
 \dfrac{\partial^2
W_m(\bt,\bb)}{\partial\beta_i\,\partial\beta_j}=0,\quad \forall i\neq j.
 \end{equation}
Given  $2\leq r\leq m$, a solution $(\bt,\bb)$ of \eqref{s0} such that exactly $r$ of its components are equal will be called a
$r$-\emph{reduced solution} of \eqref{s0}.

The formulation  \eqref{hod} of the hodograph equations for the dcKdV$_m$ hierarchies allows us
to apply the theory of critical points of functions to analyze the solutions of these hierarchies, while
 \eqref{epd} indicates that the functions $W_m$ are of a very special class.

The EPD equation \eqref{epd} arose in the study of cyclids \cite{dar},  where solutions $W$
of the above form have been found too. Much later it appeared in the theory of Whitham equations describing the
small dispersion limit of the KdV equation \cite{tia1,gra}.

We note that hodograph equations of a form close to \eqref{hod} have been presented in \cite{pav} and \cite{dubk}.
 Furthermore, linear equations of the EPD type  and their connection with hydrodynamic chains
 have been studied in \cite{pav1} too.

Finally, we emphasize that the functions $W_m$ depend on the parameters $t_1,t_2,\ldots$ (times of the hierarchy).
Since "degenerate critical points appear naturally in cases when the functions depend on parameters " \cite{arn1,arn2}, one should expect the existence of families of degenerate critical points for the functions $W_m$. Their connection with the singular sectors in the spaces of solutions for dcKdV$_m$ will be considered in the next section.

\vspace{0.3cm}

To illustrate the statements given above we next present some simple examples. For the dcKdV$_2$ hierarchy we have
\begin{align*}
 W_2(\bt,\bb)&=\dfrac{x}{2}(\beta_1+\beta_2)+\dfrac{t_1}{8}(3\beta_1^2+2\beta_1\beta_2+3\beta_2^2)+\dfrac{t_2}{16}
 \left(5\beta_1^3+
  3\beta_1^2\beta_2+3\beta_1\beta_2^2+5\beta_2^3\right)\\
  &+
  \dfrac{t_3}{128}(35\beta_1^4+20\beta_1^3\beta_2+18\beta_1^2\beta_2^2+20\beta_1\beta_2^3+35\beta_2^4)+\cdots
\end{align*}
The hodograph equations with $t_n\,=\,0$ for  $n\,\geq\,4$,  take the
form
\begin{equation}\label{h2}
\everymath{\displaystyle}
\begin{cases}
8x+4t_1(3\beta_1+\beta_2)+3t_2\left(5\beta_1^2+
  2\beta_1\beta_2+\beta_2^2\right)+
  \dfrac{t_3}{8}(140\beta_1^3+60\beta_1^2\beta_2+36\beta_1\beta_2^2+20\beta_2^3)=0,\\\\
8x+4t_1(\beta_1+3\beta_2)+3t_2\left(\beta_1^2+
  2\beta_1\beta_2+5\beta_2^2\right)
  +\dfrac{t_3}{8}(140\beta_2^3+60\beta_2^2\beta_1+36\beta_2\beta_1^2+20\beta_1^3)=0.
\end{cases}
\end{equation}

\vspace{0.3cm}

 For the dcKdV$_3$ hierarchy we have
\begin{align*}\everymath{\displaystyle}
W_3(\bt,\bb)&=\,\frac{x}{2}(\beta_1+\beta_2+\beta_3)+\frac{t_1}{8}\left(3 \beta_1^2+3 \beta_2^2+3 \beta_3^2+2\beta_1\beta_2+2\beta_1\beta_3
   +2 \beta_2\beta_3\right)\\  \\
   &+\frac{t_2}{16} \Big(5 \beta_1^3++5 \beta_2^3+5 \beta_3^3+3\beta_1^2\beta_2+3\beta_1^2\beta_3+3\beta_1\beta_2^2+3\beta_2^2\beta_3+
                          3\beta_1\beta_3^2\\  \\
   &\quad +3\beta_2\beta_3^2+2\beta_1\beta_2 \beta_2\Big)+\cdots
\end{align*}
The hodograph equations with $t_n\,=\,0$ for $n\,\geq\,3$  are
\begin{equation}\label{h3}
\everymath{\displaystyle}\begin{cases}
8\,x+4\,t_1\,(3\,\beta_1+\beta_2+\beta_3)+t_2\,(15\,\beta_1^2+3\,\beta_2^2+3\,\beta_3^2
+6\,\beta_1\,\beta_2+6\,\beta_1\,\beta_3+2\,\beta_2\,\beta_3)=0,\\  \\
8\,x+4\,t_1\,(\beta_1+3\,\beta_2+\beta_3)+t_2\,(3\,\beta_1^2+15\,\beta_2^2+3\,\beta_3^2
+6\,\beta_1\,\beta_2+2\,\beta_1\,\beta_3+6\,\beta_2\,\beta_3)=0,\\  \\
8\,x+4\,t_1\,(\beta_1+\beta_2+3\,\beta_3)+t_2\,(3\,\beta_1^2+3\,\beta_2^2+15\,\beta_3^2
+2\,\beta_1\,\beta_2+6\,\beta_1\,\beta_3+6\,\beta_2\,\beta_3)\,=\,0.
\end{cases}
\end{equation}

\section{Singular sectors of dcKdV$_m$  hierarchies}

We say that $(\bt,\bb)\in \mathcal{M}_m$ is a regular point if it is a nondegenerate critical point of the function
$W_m$. That it is to say, if it satisfies \cite{arn1,arn2}
\begin{equation}\label{reg}
 \mbox{det} \Big(\dfrac{\partial^2 W_m(\bt,\bb)}{\partial \beta_i\, \partial \beta_j}\Big)\neq 0.
 \end{equation}
The set of regular points of $\mathcal{M}_m$ will be denoted by $\mathcal{M}_m^{\mbox{reg}}$ and the points of its
complementary set $\mathcal{M}_m^{\mbox{sing}}:=\mathcal{M}_m-\mathcal{M}_m^{\mbox{reg}}$, where the second differential of $W_m$
is a degenerate quadratic form, will be called singular points. We will also refer to $\mathcal{M}_m^{\mbox{reg}}$
and $\mathcal{M}_m^{\mbox{sing}}$ as the regular and singular sectors of the  dcKdV$_m$  hierarchy. So $\mathcal{M}_m^{\mbox{sing}}$
describes families of degenerate critical points of the function $W_m$. Near a  regular point the variety $\mathcal{M}_m^{\mbox{reg}}$ can be uniquely described as $(\bt,\bb(\bt))$ where $\bb(\bt)$ is a solution of the dcKdV$_m$  hierarchy.

The aim of this section is to analyze the structure of $\mathcal{M}_m^{\mbox{sing}}$ by taking advantage of the special properties of the set of coordinates $\bb$.

 In general, the singular sectors of dcKdV$_m$ hierarchies with $m\geq 2$ contain both reduced and unreduced points.  For example, the hodograph equations \eqref{h2} of the dcKdV$_2$ hierarchy have reduced singular points
given by $(x,t_1,t_2,t_3,\beta_1=\beta_2)$  where
\[
72xt_3^2=-9t_2^2+36t_1t_2t_3+(8t_1t_3-3t_2^2)\sqrt{9t_2^2-24t_1t_3},
\]
and
\[
\beta_1=\beta_2=-\dfrac{3t_2+\sqrt{9t_2^2-24t_1t_3}}{12t_3}.
\]
Furthermore, there are also unreduced singular points $(x,t_1,t_2,t_3,\beta_1,\beta_2)$ determined by the constraint
$$
360xt_3^3=-45 t_3 t_2^3+180 t_1
   t_3^2 t_2+\sqrt{15}\,(8t_1 t_3-3t_2^2)\,\sqrt{t_3^2
   \left(3 t_2^2-8 t_1 t_3\right)},$$
and
$$\beta_1\,=\,\frac{-3 t_2 t_3+\sqrt{15} \sqrt{t_3^2 \left(3 t_2^2-8 t_1
   t_3\right)}}{12 t_3^2},\quad
   \beta_2\,=\,-\frac{5 t_2 t_3+\sqrt{15} \sqrt{t_3^2 \left(3
   t_2^2-8 t_1 t_3\right)}}{20 t_3^2}$$

\vspace{0.3cm}

From \eqref{cros} it follows at once that
\begin{teh}
Let $(\bt,\bb)$ be an unreduced  solution of the hodograph equations \eqref{s0}, then $(\bt,\bb)$ is a singular point
if and only if at least one of the derivatives
\[
\dfrac{\partial^2\,W_m(\bt,\bb)}{\partial\,\beta_i ^2},\quad i=1,\ldots,m,
\]
 vanishes.
\end{teh}

Notice that since the function $W_m$ satisfies the EPD equation \eqref{epd} ,  its partial derivatives  at unreduced points $(\bt,\bb)$
\[
\dfrac{\partial^{q} W_m(\bt,\bb)}{\partial \beta_1^{q_1}\cdots\partial\beta_m^{q_m} },\quad q:=q_1+\cdots+q_m,
\]
can always be  expressed as a linear combination of  \emph{diagonal} derivatives $\partial_{\beta_i}^{k_i}W_m$ with $k_i\leq q_i$.
Thus, for each vector  $\bq=(q_1,\ldots,q_m)\in \N^m$ with at least one $q_i\geq 1$ it is natural to introduce an associated subvariety
 $\mathcal{M}_{m,\bq}^{\mbox{sing}}$ of $\mathcal{M}_{m}^{\mbox{sing}}$  defined as the set of unreduced
solutions  $(\bt,\bb)$ of the hodograph equations \eqref{s0}  such that
\begin{equation}\label{classso}
\dfrac{\partial^{k_i}\,W_m(\bt,\bb)}{\partial\,\beta_i ^{k_i}}=0,\quad \forall  k_i\leq q_i+1.
\end{equation}

In particular, for $\bq=(0,\ldots,0,q)$ with $q\geq 1$ we denote by $\mathcal{M}_{m,q}^{\mbox{sing}}$  the subvariety  associated to
$\bq=(0,\ldots,0,q)$. That is to say,  $\mathcal{M}_{m,q}^{\mbox{sing}}$  is the set of solutions  $(\bt,\bb)$ of the hodograph equations \eqref{s0}  such that
\begin{equation}\label{classs}
\dfrac{\partial^2\,W_m(\bt,\bb)}{\partial\,\beta_m ^2}=\dfrac{\partial^3\,W_m(\bt,\bb)}{\partial\,\beta_m ^3}=\ldots=
\dfrac{\partial^{q+1}\,W_m(\bt,\bb)}{\partial\,\beta_m ^{q+1}}=0.
\end{equation}
These  subvarieties    define   a nested sequence
\begin{equation}\label{fil}
\mathcal{M}_{m}^{\mbox{sing}}\supset  \mathcal{M}_{m,1}^{\mbox{sing}} \supset \mathcal{M}_{m,2}^{\mbox{sing}}\supset \cdots\mathcal{M}_{m,q}^{\mbox{sing}}\supset \cdots,
\end{equation}
and represent sets of points whose  singular degree increases with $q$.
Moreover, due to the covariance of  the functions $F_i=\partial_{\beta_i}\,W_m$
under permutations there is no need of introducing alternative sequences of the form \eqref{classs}  based on systems of equations   corresponding to the remaining coordinates $\beta_j$ for $j\neq m$.

The next result states that the varieties $ \mathcal{M}_{m,q}^{\mbox{sing}}$ of the dcKdV$_m$  hierarchy are closely related to the $(2\,q+1)$-reduced solutions of the dcKdV$_{m+2\,q}$ hierarchy.

\vspace{0.3cm}

Notice that given $2\leq r\leq m$, the hodograph equations for $r$-reduced solutions
\[
\beta_{m-r+1}=\beta_{m-r+2}=\ldots=\beta_m,
\]
of the dcKdV$_m$  hierarchy reduce to the  system
\[
F_i(\bt,\bb)=0,\quad i=1,\ldots,m-r+1,
\]
of $m-r+1$ equations for the $m-r+1$ unknowns $(\beta_1,\ldots,\beta_{m-r+1})$. Now we prove

\begin{teh}
If $(\bt,\bb)\in\mathcal{M}_{m,q}^{\mbox{sing}}$ where $\bt=(t_0,t_1,\ldots)$ and $\bb=(\beta_1,\ldots,\beta_m)$, then if we define
\[
\bt^{(m+2\,q)}:=(t_q,t_{q+1},\ldots),\quad
\bb^{(m+2\,q)}:=(\beta_1,\ldots,\beta_m,\overbrace{\beta_m,\ldots,\beta_m}^{2\,q}),
\]
it follows that $
(\bt^{(m+2\,q)},\bb^{(m+2\,q)})
$ is a $(2\,q+1)$-reduced solution of the hodograph equations for the dcKdV$_{m+2\,q}$ hierarchy.
\end{teh}
\begin{proof}
To proof this statement we will use superscripts $(m)$ and $(m+2\,q)$ to distinguish objects corresponding to different  hierarchies. By assumption we have that
\[
(\bt^{(m)},\bb^{(m)})\in\mathcal{M}_{m,q}^{\mbox{sing}}.
\]
Thus, taking \eqref{idd} into account, we have that \eqref{classs} can be rewritten as
\begin{equation}\label{idd1}\everymath{\displaystyle}
\begin{cases}
F_i^{(m)}(\bt^{(m)},\bb^{(m)}):=\oint_{\gamma}\dfrac{\d \lambda}{2\,i\,\pi}\,\dfrac{U^{(m)}(\lambda,\bt^{(m)})\,
R^{(m)}(\lambda,\bb^{(m)})}{\lambda-\beta_i^{(m)}}=0,\quad i=1,\ldots,m \\\\
F_{m,j}^{(m)}(\bt^{(m)},\bb^{(m)}):=\oint_{\gamma}\dfrac{\d \lambda}{2\,i\,\pi}\,\dfrac{U^{(m)}(\lambda,\bt^{(m)})\,
R^{(m)}(\lambda,\bb^{(m)})}{(\lambda-\beta_m^{(m)})^j}=0,\quad j=2,\ldots,q+1.
\end{cases}
\end{equation}
Now a $(2\,q+1)$-reduced solution  of the hodograph equations for  the dcKdV$_{m+2\,q}$ is characterized by
\begin{equation}
F_i^{(m+2\,q)}(\bt^{(m+2\,q)},\bb^{(m+2\,q)}):=\oint_{\gamma}\dfrac{\d \lambda}{2\,i\,\pi}\,\dfrac{U^{(m+2\,q)}(\lambda,\bt^{(m+2\,q)})\,
R^{(m+2\,q)}(\lambda,\bb^{(m+2\,q)})}{\lambda-\beta_i^{(m+2\,q)}}=0,\end{equation}
where $i=1,\ldots,m
$.
But it is clear that
\begin{equation}\label{idd3}
R^{(m+2\,q)}(\lambda,\bb^{(m+2\,q)})=\dfrac{\lambda^q}{(\lambda-\beta_m^{(m)})^q}\,R^{(m)}(\lambda,\bb^{(m)})
\end{equation}
Hence if we set
\[
t^{(m+2\,q)}_i:=t^{(m)}_{i+q},\quad i\geq 0,
\]
we have
\begin{equation}\label{idd4}
U^{(m)}(\lambda,\bt^{(m)})=x^{(m)}+\lambda\,t_1^{(m)}+\cdots+
\lambda^{q-1}\,t_{q-1}^{(m)}+\lambda^q\,
U^{(m+2\,q)}(\lambda,\bt^{(m+2\,q)}).
\end{equation}
Then it follows that
\begin{equation}\label{idd4a}
F_i^{(m+2\,q)}(\bt^{(m+2\,q)},\bb^{(m+2\,q)})=
\oint_{\gamma}\dfrac{\d \lambda}{2\,i\,\pi}\,\dfrac{U^{(m)}(\lambda,\bt^{(m)})\,
R^{(m)}(\lambda,\bb^{(m)})}{(\lambda-\beta_i^{(m)})(\lambda-\beta_m^{(m)})^q}, \quad i=1,\ldots,m.
\end{equation}

Furthermore, for any given $i=1,\ldots,m$  we have
\begin{align*}
& F_{i}^{(m)}(\bt^{(m)},\bb^{(m)})=\oint_{\gamma}\dfrac{\d \lambda}{2\,i\,\pi}\,\dfrac{U^{(m)}(\lambda,\bt^{(m)})\,
R^{(m)}(\lambda,\bb^{(m)})}{\lambda-\beta_i^{(m)}}\\\\
&=\oint_{\gamma}\dfrac{\d \lambda}{2\,i\,\pi}\,
\dfrac{(\lambda-\beta_m^{(m)})^{q}\,U^{(m)}(\lambda,\bt^{(m)})\,
R^{(m)}(\lambda,\bb^{(m)})}{(\lambda-\beta_i^{(m)})\,(\lambda-\beta_m^{(m)})^q}\\\\
&=\sum_{k=0}^{q}\,c_{1,k}(\bb^{(m)})\,I_{i,k}(\bt^{(m)},\bb^{(m)}) ,
\end{align*}
and
\begin{align*}
& F_{m,j}^{(m)}(\bt^{(m)},\bb^{(m)})=\oint_{\gamma}\dfrac{\d \lambda}{2\,i\,\pi}\,\dfrac{U^{(m)}(\lambda,\bt^{(m)})\,
R^{(m)}(\lambda,\bb^{(m)})}{(\lambda-\beta_m^{(m)})^j}\\\\
&=\oint_{\gamma}\dfrac{\d \lambda}{2\,i\,\pi}\,
\dfrac{(\lambda-\beta_i^{(m)})\,(\lambda-\beta_m^{(m)})^{q-j}\,U^{(m)}(\lambda,\bt^{(m)})\,
R^{(m)}(\lambda,\bb^{(m)})}{(\lambda-\beta_i^{(m)})\,(\lambda-\beta_m^{(m)})^q}\\\\
&=\sum_{k=0}^{q-j+1}\,c_{j,k}(\bb^{(m)})\,I_{i,k}(\bt^{(m)},\bb^{(m)})
,\quad j=2,\ldots, q+1.
\end{align*}
where the functions $c_{j,k}(\bb^{(m)})$ are the coefficients of the polynomials
\begin{equation}\label{pol}
\begin{cases}
(\lambda-\beta_m^{(m)})^q=\sum_{k=0}^q c_{1k}(\bb^{(m)})\,\lambda^k;\\\\
(\lambda-\beta_i^{(m)})\,(\lambda-\beta_m^{(m)})^{q-j}=\sum_{k=0}^{q-j+1} c_{jk}(\bb^{(m)})\,\lambda^k,\quad j=2,\ldots,q+1.
\end{cases}
\end{equation}
and
\begin{equation}\label{int}
I_{i,k}(\bt^{(m)},\bb^{(m)}):=\oint_{\gamma}\dfrac{\d \lambda}{2\,i\,\pi}\,
\dfrac{\lambda^k\,U^{(m)}(\lambda,\bt^{(m)})\,
R^{(m)}(\lambda,\bb^{(m)})}{(\lambda-\beta_i^{(m)})\,(\lambda-\beta_m^{(m)})^q}
\end{equation}
Now, for any given $i=1,\ldots,m$  the system \eqref{classs} implies
\[
\begin{cases}
F_{i}^{(m)}(\bt^{(m)},\bb^{(m)})=0,\\\\
F_{m,j}^{(m)}(\bt^{(m)},\bb^{(m)})=0,\quad j=2,\ldots, q+1,
\end{cases}
\]
and, as a consequence, we deduce the following system of
$q$ homogeneous linear equations for the $q$ functions $I_{i,k}(\bt^{(m)},\bb^{(m)})$
\[
\sum_{k=0}^{q-j+1}\,c_{j,k}(\bb^{(m)})\,I_{i,k}(\bt^{(m)},\bb^{(m)})=0,
\quad j=1,\ldots, q+1.
\]
Because of the linear independence of the polynomials \eqref{pol} these equations are linearly independent and, therefore, all
the functions $I_{i,k}(\bt^{(m)},\bb^{(m)})$ vanish. Finally, from \eqref{idd4a}
we conclude that $I_{i,0}(\bt^{(m)},\bb^{(m)})=0$ is equivalent to $F_i^{(m+2\,q)}(\bt^{(m+2\,q)},\bb^{(m+2\,q)})=0$ and the statement follows.
\end{proof}

\section{Examples}

\subsubsection*{ dcKdV$_1$ hierarchy }

The hodograph equation  for the dcKdV$_1$ hierarchy with
$t_n=0$ for all $n\geq3$ reduce to
\begin{equation}\label{h00}
8\,x+12\,t_1\,\beta_1+15\,t_2\,\beta_1^2=0.
\end{equation}
The singular variety $\mathcal{M}_{1,1}^{\mbox{sing}}$ for
\eqref{h00} is determined by adding to \eqref{h00} the equation
\begin{equation}\label{ex02}
2\,t_1+5\,t_2\,\beta_1=0,
\end{equation}
so that  for $t_2\neq  0$ we have $\beta_1=-\,\dfrac{2\,t_1}{5\,t_2}$.
Substituting this result in  \eqref{h00} we find a constraint for
the flow parameters
\[
x=\dfrac{3}{10}\,\dfrac{t_1^2}{t_2},
\]
which is  the shock region for the solution of \eqref{h00} given by
\begin{equation}\label{sol0}
\beta_1=\dfrac{2}{15\,t_2}\,\Big(-3\,t_1+\sqrt{3\,(3\,t_1^2-10\,t_2\,x)}\,\Big).
\end{equation}

There are two sectors $\mathcal{M}_{1,1,k}^{\mbox{sing}} \,(k=1,2)$
in $\mathcal{M}_{1,1}^{\mbox{sing}}$
\begin{align}\label{ex03}
\nonumber &\mathcal{M}_{1,1,1}^{\mbox{sing}}:\quad x=t_1=t_2=0,\quad
\mbox{$\beta_1$ arbitrary};\\\\
\nonumber & \mathcal{M}_{1,1,2}^{\mbox{sing}}:\quad
\mbox{$(x,t_1,t_2,\beta_1)$ such that $t_2\neq 0$,\,
$x=\dfrac{3}{10}\,\dfrac{t_1^2}{t_2}$ and
$\beta_1=-\dfrac{2}{5}\,\dfrac{t_1}{t_2}$}
\end{align}

To see the relationship with the dcKdV$_3$  hierarchy we  notice that
\[
x^{(3)}=t_1,\quad t_1^{(3)}=t_2,
\]
and
\[
\bb^{(3}=(\beta_1,\beta_1,\beta_1)=-\dfrac{2}{5}\,\dfrac{x^{(3)}}{t_1^{(3)}}\,(1,1,1),
\]
which is a $3$-reduced solution of the first flow \eqref{f3} of the
dcKdV$_3$ hierarchy.

\vspace{0.4cm}

The dcKdV$_1$ hodograph equation  with $t_n\,=\,0$ for all $n\,\geq\,6$ is

$$693\,t_5\,\beta_1^5\,+\,630\,t_4\,\beta_1^4\,+\,560\,t_3\,\beta_1^3\,+\,480\,t_2\,\beta_1^2\,+\,384\,t_1\,\beta_1\,+256\,x\,=\,0.$$
Let us first consider the singular variety ${\cal M}^{\mbox{sing}}_{1,1}$ with $t_n\,=\,0$ for all $n\,\geq\,4$. It is
is determined by the equations
$$\everymath{\displaystyle}\begin{array}{l}
560\,t_3\,\beta_1^3\,+\,480\,t_2\,\beta_1^2\,+\,384\,t_1\,\beta_1\,+256\,x\,=\,0,\\  \\
1680\,t_3\,\beta_1^2\,+\,960\,t_2\,\beta_1\,+\,384\,t_1\,=\,0.
\end{array}$$
Thus an open subset of ${\cal M}^{\mbox{sing}}_{1,1}$ can be parametrized by the equations
$$\everymath{\displaystyle}\begin{array}{l}
x\,=\,\frac{-\,25\,t_2^3\,+\,105\,t_1\,t_2\,t_3\,+\,\sqrt{5}\,\sqrt{125\,t_2^6\,-\,1050\,t_1\,t_3\,t_2^4\,+\,2940\,t_1^2\,t_2^2\,t_3^2\,
-\,2744\,t_1^3\,t_3^3}}{245\,t_3^2},\\  \\
\beta_1\,=\,-\,\frac{2\,\left(-\,25\,t_2^3\,+\,70\,t_1\,t_2\,t_3\,+\,\sqrt{5}\,
   \sqrt{\left(5\,t_2^2\,-\,14\,t_1\,t_3\right)^3}\right)}{35\,t_3\,\left(14\,t_1\,t_3\,-\,5\,t_2^2\right)}.
\end{array}$$
It determines the following 3-reduced solution of the two first flows of  the dcKdV$_3$  hierachy
($x^{(3)}\,=\,t_1$, $t_1^{(3)}\,=\,t_2$, $t_2^{(3)}\,=\,t_3$)
$$\beta_1^{(3)}=\beta_2^{(3)}=\beta_3^{(3)}\,=\,-\frac{2 \left(-25\,(t_1^{(3)})^3\,+\,70\,x^{(3)}\,t_1^{(3)}\,t_2^{(3)}\,+\,\sqrt{5}\,\sqrt{\left(5\,(t_1^{(3)})^2\,
-\,14\,x^{(3)}\,t_2^{(3)}\right)^3}\right)}{35\,t_2^{(3)}\,\left(14\,x^{(3)}\,t_2^{(3)}\,-\,5\,(t_1^{(3)})^2\right)}.$$

Next, for the sector ${\cal M}^{\mbox{sing}}_{1,2}$ if we set $t_n\,=\,0$ for all $n\,\geq\,5$,  we obtain the equations
$$\everymath{\displaystyle}\begin{array}{l}
630\,t_4\,\beta_1^4\,+\,560\,t_3\,\beta_1^3\,+\,480\,t_2\,\beta_1^2\,+\,384\,t_1\,\beta_1\,+256\,x\,=\,0,\\  \\
2520\,t_4\,\beta_1^3\,+\,1680\,t_3\,\beta_1^2\,+\,960\,t_2\,\beta_1\,+\,384\,t_1\,=\,0,\\  \\
7560\,t_4\,\beta_1^2\,+\,3360\,t_3\,\beta_1\,+\,960\,t_2\,=\,0.
\end{array}$$
From these equations we find
$$\everymath{\displaystyle}\begin{array}{l}
t_1\,=\,\frac{5\,\left(-\,49\,t_3^3\,+\,189\,t_2\,t_3\,t_4\,+\,\sqrt{7}\,\sqrt{343\,t_3^6\,-\,2646\,t_2\,t_4\,t_3^4\,+\,
6804\,t_2^2\,t_3^2\,t_4^2\,-\,5832\,t_2^3\,t_4^3}\right)}{1701\,t_4^2},\\  \\
x\,=\,\frac{5\,\left(-\,98\,t_3^4\,+\,378\,t_2\,t_4\,t_3^2\,+\,2\,\sqrt{7}\,\sqrt{\left(7\,t_3^2\,-\,18\,t_2\,t_4\right)^3}
\,t_3\,-\,243\,t_2^2\,t_4^2\right)}{10206\,t_4^3},\\  \\
\beta_1\,=\,-\,\frac{2\,\left(-\,49\,t_3^3\,+\,126\,t_2\,t_3\,t_4\,+\,\sqrt{7}\,\sqrt{\left(7\,t_3^2\,-\,18\,t_2\,t_4\right)^3}
\right)}{63\,t_4\,\left(18\,t_2\,t_4\,-\,7\,t_3^2\right)}.
\end{array}$$
Then the associated 5-reduced solution of the two first flows of the dcKdV$_5$  hierarchy
($x^{(5)}\,=\,t_2$, $t_1^{(5)}\,=\,t_3$, $t_2^{(5)}\,=\,t_4)$ is given by
$$\beta_i\,=\,-\frac{2\,\left(-\,49\,(t_1^{(5)})^3\,+\,126\,x^{(5)}\,t_1^{(5)}\,t_2^{(5)}\,+\,\sqrt{7}\,
\sqrt{\left(7\,(t_1^{(5)})^2\,-\,18\,x^{(5)}\,t_2^{(5)}\right)^3}\right)}{63\,t_2^{(5)}\,\left(18\,x^{(5)}\,t_2^{(5)}\,-\,
7\,(t_1^{(5)})^2\right)},\quad i\,=1,\dots,5.$$

\subsubsection*{ dcKdV$_2$ hierarchy }

Let us consider the hodograph equations for the  dcKdV$_2$
hierarchy with $t_n=0$ for all $n\geq3$. From \eqref{h2} we have that they take the form
\begin{equation}\label{h22}
\everymath{\displaystyle}
\begin{cases}
8x+4t_1(3\beta_1+\beta_2)+3t_2\left(5\beta_1^2+
  2\beta_1\beta_2+\beta_2^2\right)=0,\\\\
8x+4t_1(\beta_1+3\beta_2)+3t_2\left(\beta_1^2+
  2\beta_1\beta_2+5\beta_2^2\right)
  =0.
\end{cases}
\end{equation}
The singular variety $\mathcal{M}_2^{\mbox{sing}}$ is determined by \eqref{h22} together with the additional
condition \newline  ($\mbox{det} (\partial_{\beta_i\beta_j} W_m(\bt,\bb))=0$)
\begin{equation}\label{add}
-(2t_1+3t_2(\beta_1+\beta_2))^2+9(2t_1+t_2(5\beta_1+\beta_2))(2t_1+t_2(\beta_1+5\beta_2))=0.
\end{equation}
There elements  of  $\mathcal{M}_{2}^{\mbox{sing}} $ are
\begin{align}\label{ex3}
\nonumber &x=t_1=t_2=0,\quad
\mbox{$(\beta_0,\beta_1)$ arbitrary};\\\\
\nonumber &
\mbox{$(x,t_1,t_2,\beta_1,\beta_2)$ such that $t_2\neq 0$,\,
$x=\dfrac{t_1^2}{3\,t_2}$ and $\beta_1=\beta_2=-\dfrac{t_1}{3\,t_2}$}
\end{align}
The subvarieties ${\cal M}^{\mbox{sing}}_{2,q}$ are all equal and given by
\[
x=t_1=t_2=0,\quad \mbox{$(\beta_0,\beta_1)$ arbitrary with $\beta_0\neq \beta_1$}.
\]

Notice that the constraint $x=\dfrac{t_1^2}{3\,t_2}$ determines  the shock region for the following solution of \eqref{h22}
\begin{equation}\label{h2a}
\beta_1\,=\,\frac{-\,t_1\,+\,\sqrt{2}\,\sqrt{t_1^2\,-\,3\,t_2\,x}}{3\,t_2},\quad
\beta_2\,=\,\frac{-\,t_1\,-\,\sqrt{2}\,\sqrt{t_1^2\,-\,3\,t_2\,x}}{3\,t_2}.
\end{equation}

Let us now consider the system of hodograph equations \eqref{h2} for
the  dcKdV$_2$ hierarchy with $t_n=0$ for all $n\geq4$. The singular
variety $\mathcal{M}_2^{\mbox{sing}}$ is now determined by
\eqref{h2} and the condition ($\mbox{det} (\partial_{\beta_i\beta_j}
W_m(\bt,\bb))=0$)
$$\everymath{\displaystyle}\begin{array}{l}
32 t_1^2+96 t_2 (\beta_1+\beta_2) t_1+702 t_3^2 \beta_1^2 \beta_2^2+72
   \left(3 t_2^2+t_1 t_3\right) \beta_1
   \beta_2+12 \left(3 t_2^2+13 t_1 t_3\right)
   \left(\beta_1^2+\beta_2^2\right)+\\  \\
   \quad 486 t_2 t_3   \left(\beta_2 \beta_1^2+\beta_2^2 \beta_1\right)+90 t_2 t_3 \left(\beta_1^3+
   \beta_2^3\right)+180 t_3^2 \left(\beta_2 \beta_1^3+\beta_2^3 \beta_1\right)+45 t_3^2
   \left(\beta_1^4+\beta_2^4\right)=0.
\end{array}$$
One finds the following six sectors in $\mathcal{M}_2^{\mbox{sing}}$
\begin{align*}
\everymath{\displaystyle}
\textbf{1.}\quad  & x\,=\,\frac{-9 t_2^3+36 t_1 t_3
t_2+(8t_1t_3-3t_2^2)\sqrt{9t_2^2-24 t_1
   t_3}}{72t_3^2},\qquad
   \beta_1\,=\,\beta_2\,=\,-\frac{3 t_2+\sqrt{9 t_2^2-24
   t_1 t_3}}{12 t_3},\\  \\  \\
\textbf{2.}\quad  & x\,=\,\frac{-9 t_2^3+36 t_1 t_3
t_2-(8t_1t_3-3t_2^2)\sqrt{9t_2^2-24 t_1
   t_3}}{72t_3^2},\qquad
   \beta_1\,=\,\beta_2\,=\,\frac{-3 t_2+\sqrt{9 t_2^2-24
   t_1 t_3}}{12 t_3},\\ \\ \\
\textbf{3.} &\quad x\,=\,\frac{-45t_3 t_2^3+180 t_1 t_3^2
   t_2+\sqrt{15}(8t_1t_3-3t_2^2)\sqrt{t_3^2 \left(3 t_2^2-8
   t_1 t_3\right)}}{360 t_3^3},\\  \\
 &\beta_1\,=\,-\frac{5 t_2 t_3+\sqrt{15} \sqrt{t_3^2 \left(3
   t_2^2-8 t_1 t_3\right)}}{20t_3^2},\qquad \beta_2\,=\,\frac{-3 t_2 t_3+\sqrt{15} \sqrt{t_3^2 \left(3 t_2^2-8 t_1
   t_3\right)}}{12 t_3^2},\\  \\  \\
\textbf{4.} &\quad x\,=\,\frac{-45t_3 t_2^3+180 t_1 t_3^2
   t_2-\sqrt{15}(8t_1t_3-3t_2^2)\sqrt{t_3^2 \left(3 t_2^2-8
   t_1 t_3\right)}}{360 t_3^3},\\  \\
&\beta_1\,=\,-\frac{3 t_2 t_3+\sqrt{15} \sqrt{t_3^2 \left(3 t_2^2-8 t_1
   t_3\right)}}{12 t_3^2},\qquad \beta_2\,=\,\frac{-5 t_2 t_3+\sqrt{15} \sqrt{t_3^2 \left(3
   t_2^2-8 t_1 t_3\right)}}{20t_3^2},
\end{align*}
\begin{align*}
\everymath{\displaystyle}
\textbf{5.} &\quad x\,=\,\frac{-45t_3 t_2^3+180 t_1 t_3^2
   t_2-\sqrt{15}(8t_1t_3-3t_2^2)\sqrt{t_3^2 \left(3 t_2^2-8
   t_1 t_3\right)}}{360 t_3^3},\\  \\
&\beta_1\,=\,\frac{-5 t_2 t_3+\sqrt{15} \sqrt{t_3^2 \left(3
   t_2^2-8 t_1 t_3\right)}}{20t_3^2},           \,\qquad \beta_2
\,=\,-\frac{3 t_2 t_3+\sqrt{15} \sqrt{t_3^2 \left(3 t_2^2-8
t_1   t_3\right)}}{12 t_3^2}
\\  \\
\textbf{6.} &\quad x\,=\,\frac{-45t_3 t_2^3+180 t_1 t_3^2
   t_2+\sqrt{15}(8t_1t_3-3t_2^2)\sqrt{t_3^2 \left(3 t_2^2-8
   t_1 t_3\right)}}{360 t_3^3},\\  \\
   &\beta_1\,=\,\frac{-3 t_2 t_3+\sqrt{15} \sqrt{t_3^2 \left(3 t_2^2-8 t_1
   t_3\right)}}{12 t_3^2},\qquad \beta_2\,=\,-\frac{5 t_2 t_3+\sqrt{15} \sqrt{t_3^2 \left(3
   t_2^2-8 t_1 t_3\right)}}{20t_3^2}.
\end{align*}

It is easy to see that $\mathcal{M}_{2,1}^{\mbox{sing}}$ is given by
the sectors 5 and 6. To check the connection between these sectors
and  the dcKdV$_4$ hierarchy it is enough to set
\[
x^{(4)}=t_1,\quad t_1^{(4)}=t_2,\quad t_2^{(4)}=t_3,\quad
\bb^{(4)}=(\beta_1,\beta_2,\beta_2,\beta_2),
\]
and it is immediate to prove that $\bb^{(4)}(\bt^{(4)})$ verifies
the equations of the first flow of the dcKdV$_4$ hierarchy
\[
\dfrac{\partial \,\beta_i}{\partial t_1^{(4)}}=\Big(\beta_i+\dfrac{1}{2}\,\sum_{k=1}
^4 \beta_k\Big)\,\dfrac{\partial \,\beta_i}{\partial x^{(4)}},\quad i=1,\ldots,4.
\]

\vspace{0.5cm}

\subsection* { Acknowledgements}

\vspace{0.3cm} The authors  wish to thank the  Spanish Ministerio de
Educaci\'on y Ciencia (research project FIS2008-00200/FIS) for its
finantial support. B. K. is thankful to the Departamento de F\'{i}sica Te\'orica II for the kind hospitality.


\begin{thebibliography}{99}

\bibitem{kod} Y. Kodama and B.G. Konopelchenko, J. Phys. A: Math.
Gen. {\bf 35}, L489-L500 (2002)

\bibitem{bor} B.G. Konopelchenko and L. Mart\'{\i}nez Alonso, J. Phys. A: Math.
Gen. {\bf 37}, 7859 (2004)



\bibitem{whi} G. B. Whitham, \emph{Linear and nonlinear waves}, Wiley-Interscience, New York (1976)

\bibitem{dubn} B. A. Dubrovin and S. P. Novikov,
\emph{Hydrodynamics of weakly deformed soliton lattices. Differential geometry and Hamiltonian theory}, Russian Math. Surveys {\bf 44}, 35 (1989)

\bibitem{zak} V. E. Zakharov, Func. Anal. Appl. {\bf 14}, 89 (1980)

\bibitem{lma} L. Martinez Alonso, J. Math. Phys. {\bf 21}, 2342 (1980)

\bibitem{for1} M. Antonowicz and A. P. Fordy, Phys. D {\bf 28}, no. 3, 345 (1987)

\bibitem{for2} M. Antonowicz and A. P. Fordy, J. Phys. A  : Math. Gen. {\bf 21}, L269 (1988)

\bibitem{dub1} B. Dubrovin, S. Liu and Y. Zhang,
Commun. Pure. Appl. Math. {\bf 59} 559 (2006)

\bibitem{dub2} B. Dubrovin, S. Liu and Y. Zhang,
Commun. Math. Phys. {\bf 267} 117 (2006)



\bibitem{man} M. Ma\~{n}as, L. Mart\'{\i}nez Alonso and E. Medina, J. Phys. A : Math. Gen.
{\bf 30}, 4815 (1997)

\bibitem{borl} B. Konopelchenko and L. Mart\'{\i}nez Alonso and E. Medina, J. Phys. A : Math. Gen.
{\bf 32}, 3621 (1997)

\bibitem{fer} E. V. Ferapontov and M. V. Pavlov, Physica D {\bf 52},
211 (1991)



\bibitem{kri} I. M. Krichever, Commun. Pure. Appl. Math. {\bf 47}  437 (1994)


\bibitem{dar} G. Darboux, \emph{Lecons sur la theorie general des surfaces II }, Gauthier Villars (1915)

\bibitem{kud} V. R. Kudashev and S. E. Sharapov, Phys. Lett. A {\bf
154},445 (1991); Theor. Math. Phys. {\bf 87}, 40 (1991)

\bibitem{tia1}  F. R. Tian, Commun. Pure. Appl. Math. {\bf 46} 1093 (1993)

\bibitem{tia2}  F. R.  Tian, Duke Math. J. {\bf 74}  203 (1994)

\bibitem{gra}  T. Grava, Commun. Pure. Appl. Math. {\bf 55} 395 (2002)

\bibitem{pav} M. V. Pavlov, J. Math. Phys. {\bf 44} 4134 (2003)

\bibitem{pav1} M. V. Pavlov,  \emph{Hamiltonian formulation of electroforesis equations. Integrable hydrodynamic equations } Preprint, Landau Inst. Theor. Phys., Chernogolovsca (1987)

\bibitem{dubk} B. Dubrovin, T. Grava and C. Klein, J. Nonlinear Science {\bf 19} 57 (2009)


\bibitem{arn1} V. I. Arnold,Func. Anal. Appl. {\bf 6} no.4, 3 (1972) ; Russian Math. Surveys {\bf 29} no. 2,
10 (1974); Russian Math. Surveys {\bf 30} no. 5, 3 (1975)


\bibitem{arn2} V. I. Arnold, S. M. Gusein-Zade and A. N. Varchenko, {\emph Singularities of differentiable maps}
Birkh\"auser Boston, Inc. (1985)

\bibitem{next} B. Konopelchenko, L. Mart\'{i}nez Alonso and E. Medina, \emph{Singular sectors in the hodograph
 transform for the 1-layer Beney systems and it associated hierarchy}. In preparation


\end{thebibliography}
\end{document}